\documentclass[10pt]{article}
\usepackage{amsmath,amssymb}

\usepackage{amsthm}
\usepackage{mathrsfs}
\usepackage{amsmath}
\usepackage[english]{babel}

\newtheorem{theorem}{Theorem}[section]
\newtheorem{proposition}[theorem]{Proposition}
\newtheorem{lemma}[theorem]{Lemma}
\newtheorem{corollary}[theorem]{Corollary}

\theoremstyle{definition}
\newtheorem{definition}[theorem]{Definition}

\newtheorem{remark}[theorem]{Remark}
\newtheorem{assumption}{Assumption}[section]

\begin{document}

\title{Optimal investment under behavioural criteria in incomplete diffusion market models\thanks{The second author gratefully acknowledges support from grants of CONACYT,
Mexico and of the University of Edinburgh. Part of this work was done while the first author was affiliated with the University of Edinburgh.}}
\author{M. R\'asonyi\thanks{Alfr\'ed R\'enyi Institute of Mathematics, Hungarian Academy of Sciences, Budapest and P\'azm\'any P\'eter
Catholic University, Budapest} and J.G. Rodr\'\i guez-Villarreal\thanks{University of Edinburgh}}
\maketitle

\section{Introduction}

The most commonly accepted model for investors' preferences is expected utility theory, 
going back to \cite{bernoulli,neumann53}. According to the tenets of this theory,
an investor prefers a random return $X$ to $Y$ if $Eu(X)\geq Eu(Y)$ for some utility function
$u:\mathbb{R}\to\mathbb{R}$ that is usually assumed non-increasing and concave. 
More recently, other theories have emerged and pose new challenges to mathematics.

The present paper treats preferences of cumulative prospect theory (CPT), \cite{kt,tk}, where an ``S-shaped'' $u$
is considered (i.e. convex up to a certain point and concave from there on). Also,
distorted probability measures are applied for calculating the utility of a given position with 
respect to a (possibly random) benchmark $G$. We remark that techniques of the present paper easily carry over
to other types of preferences, too, such as rank-dependent utility \cite{quiggin} or 
acceptability indices \cite{cherny}.

The theory of optimal portfolio choice for CPT preferences is
in its infancy yet. Continuous-time studies almost always assume a complete market model,
\cite{bkp,jz,cd,campi,r12}. Only very specific types of incomplete continuous-time models have been treated
to date (finite mixtures of complete models; the case where the price is a martingale under the
physical measure;
the case where the market price of risk is deterministic), see \cite{reichlinthesis,RR11}.
In the present paper we make a step forward and consider incomplete models of a diffusion type where
the return of the investment in consideration depends on some economic factors. Our main
result asserts, under mild assumptions, 
the existence of an optimal strategy when the driving noise
of the economic factors is independent of that of the investment and the rate of return is
non-negative. The independence condition
is, admittedly, rather stringent and does not allow a leverage effect (see \cite{black}). 
We are also able to accomodate 
models of a specific type where the factor may have non-zero correlation with the investment.
We think that our 
results open the door for further generalizations. 

\section{Optimal investment model under behavioural criteria}

In this section definitions and notation related to the problem of behavioural optimal investment are presented,
based on 
\cite{Kry01}, \cite{CarRas11}. 

Unfortunately, most of the techniques developed in the literature for finding optimal policies
rely on either the Markovian nature of the problem 
or on convex duality. These are no longer applicable under behavioural criteria. For this reason we
shall consider a weak-type formulation 
of the control problem associated with optimal investment (Subsection \ref{battle}). 
Introducing a relaxation of the problem for which
results in \cite{Kry01} apply, we can prove the existence of an optimal investment strategy
(Subsection \ref{ketto}).

\subsection{The setting: market and preferences}\label{battle}

Fix a finite horizon $T>0$. We consider a financial market consisting of a risky asset, whose discounted price
 $\left(S_{t}\right)_{0\leqslant t\leqslant T}$ depends on economic factors. These factors are described by a
 $d$-dimensional 
stochastic processes $\left\{Y_t\right\}_{t\geqslant0}$. Without entering into rigorous definitions at this point,
our market model is described by the equations 
\begin{equation}\label{ekuY}
dY_t=\nu_t\left(Y_{\cdot}\right)dt+\kappa_t\left(Y_{\cdot}\right)dB_t,\,\,\,\mbox{and } \,Y_0=y, 
\end{equation}
\begin{equation}\label{eqS}
dS_t=\theta_t\left(Y_{\cdot}\right)S_tdt+\lambda_t\left(Y_{\cdot}\right)S_tdW_t\,\,\,\mbox{and }\, S_0=s>0, 
\end{equation}
with $B,W$ independent standard Brownian motions of appropriate dimensions.

We also assume that there is a riskless asset of constant price equal to $1$. We shall be more specific later in this section.
Stochastic volatility models provide prime examples of financial market models with
dynamics \eqref{ekuY} and \eqref{eqS}, see \cite{papa}.

The investor trades in the risky and riskless assets, investing a proportion $\phi_t\in\left[0,1\right]$ of his wealth into the risky asset
at time $t$. This leads to the following equation for the wealth of the investor at time $t$:
\begin{equation}\label{eqS2}
dX_t=\phi_t\theta_t\left(Y_{\cdot}\right)X_tdt+\phi_t\lambda_t\left(Y_{\cdot}\right)X_tdW_t\,\,\,\mbox{and }\, X_0=x, 
\end{equation}
where $x>0$ is the investor's initial capital.

Borrowing and short selling are not allowed, hence $\phi_t$ is a process taking values in $\left[0,1\right]$.
We note that, in this model, the risky asset's price has no influence on the economic factors. We will see in Section \ref{ext} below
how this assumption can be weakened. 
 
We will need certain closedness results on the laws of It\^{o} processes from \cite{Kry01}
hence it is necessary to work in the 'weak' setting of stochastic control theory, where the underlying probability space is not fixed.

We first set out the requirements for the coefficients in \eqref{ekuY}, \eqref{eqS2}.

Let $C\left(\left[0,T\right];\mathbb{R}^n\right)$ denote the family of $\mathbb{R}^n$-valued continuous
functions on $[0,T]$.

Denote by $p_{t}:C\left([0,T],\mathbb{R}^{d}\right)\rightarrow\mathbb{R}^d$ 
the projections $p_{t}\left(x_{\cdot}\right)=x_{t}$ and 
define the $\sigma$-algebras 
$\mathscr{N}_{t}=\sigma\left(\left\{ p_{s}:s\leqslant t\right\} \right)$,  
and $\mathscr{N}=\sigma\left(\left\{ p_{s}:s\leqslant T\right\} \right)$.

\begin{definition} \label{def2}
Let $\nu\left(t,y_{\cdot}\right):[0,T]\times C\left([0,T],\mathbb{R}^{d}\right)\to \mathbb{R}^d$
be such that the restriction of  $\nu$ to $\left[0,t\right]\times C\left(\left[0,T\right];\mathbb{R}^d\right)$ is 
$\mathcal{B}\left([0,t]\right)\otimes \mathscr{N}_{t}$-measurable, for any $0\leqslant t\leqslant T$.
We shall denote this functional by either $\nu_t\left(y_{\cdot}\right)$  or $\nu\left(t,y_{\cdot}\right)$.

Similarly, we define the coefficients $\theta, \lambda, \kappa $ with the same measurability properties 
as $\nu$, but with values in $\mathbb{R}$, $\mathbb{R}$ and $S_+^d$, respectively, where 
$S^{d}_{+}$ denotes the set of real, symmetric and positive semidefinite $d\times d$ matrices. 
\end{definition}

\begin{definition} \label{def1}
An investment stategy $\pi$ is given by the following collection: 
$$
\pi:=\left(\Omega,\mathcal{F},\left\{\mathcal{F}_t\right\}_{0\leqslant t \leqslant T},\mathbb{P},X_t,Y_t,\phi_t,\left(B_t,W_t\right),(x,y)\right),
$$ 
with $x>0$ and $y\in\mathbb{R}^d$, where
\begin{description}
\item[(a)] $\left(\Omega,\mathcal{F},\left\{\mathcal{F}_t\right\}_{0\leqslant t \leqslant T},\mathbb{P}\right)$ is a complete filtered probability space  whose filtration satisfies the usual 
conditions;
\item[(b)] the process $\left(B_t,W_t\right)_{t \geqslant 0}$ is a standard $d+1$-dimensional $\mathcal{F}_t$-Wiener process;
\item[(c)] $\phi_t:\Omega\times\left[0,T\right]\rightarrow[0,1]$ is $\mathcal{F}\otimes\mathcal{B}\left(\left[0,T\right]\right)$-measurable and 
$\mathcal{F}_t$-adapted;
\item[(d)] on the filtered probability space $X_t,Y_t$ are $\mathcal{F}\otimes\mathcal{B}\left([0,T]\right)$-measurable and 
$\mathcal{F}_t$-adapted processes such that
\begin{equation}\label{eq1}
Y_t=y+\int_0^t\nu_s\left(Y_{\cdot}\right)ds+\int_0^t \kappa_s\left(Y_{\cdot}\right)dB_s, 
\end{equation}
\begin{equation}\label{eq2}
X_t=x+\int_0^t\phi_s\theta_s\left(Y_{\cdot}\right)X_sds+\int_0^t \phi_s\lambda_s\left(Y_{\cdot}\right)X_sdW_s, 
\end{equation}
for ${0\leqslant t \leqslant T}$.
\end{description}
\end{definition}

In other words,  $\left(\Omega,\mathcal{F},\left\{\mathcal{F}_t\right\}_{0\leqslant t \leqslant T},\mathbb{P},X_t,Y_t,\left(B_t,W_t\right),(x,y)\right)$ is a 
weak solution of the system of equations \eqref{ekuY}, \eqref{eqS2}. The process $\phi_t$ represents a ratio of investment 
in the risky asset, it is measurable and $\mathcal{F}_t$-adapted.
We do not consider the price process $S_t$ from \eqref{eqS}  at all since it is enough to work with the 'controlled dynamics' $X_t$.

When needed, we will use the notation $X^{\pi},Y^{\pi}$, etc. to indicate that the object we mean
belongs to $\pi$. Let $\Pi=\Pi(x,y)$ denote the collection of all strategies.

\begin{assumption}\label{positv}
The functional $\theta$ is non-negative i.e. $\theta\left(t,y_{\cdot}\right)\geqslant 0$ for all $t\in\mathbb{R}_{+}$ and 
$y\in C\left(\left[0,T\right];\mathbb{R}^{d}\right)$.
\end{assumption}
\begin{remark}\label{pvereturn}
In other words, the return of the risky asset must be non-negative. This looks rather a harmless assumption. 
On the other hand, as mentioned before, (b) in Definition \ref{def1} is stringent. 
It excludes the 'leverage effect' where the volatility and the stock prices have (negative) correlation. This condition can be relaxed, see Section \ref{ext}.
\end{remark}

We now present the framework of optimal investment under CPT, as proposed in \cite{tk}. We follow \cite{RR11} and \cite{CarRas11}.

The investor assesses strategies by means of utilities on gains and losses, which are described in terms of functions $u_{\pm}:\mathbb{R}_{+}\rightarrow\mathbb{R}_+$, by a reference point $G$ and 
functions $w_{\pm}:\left[0,1\right]\rightarrow\left[0,1\right]$. The latter functions $w_{\pm}$ are introduced 
with the aim of explaining the 
distortions of her perception on the ``likelihood'' of her gains and losses.

According to the tenets of CPT, investors use 
benchmarks to asses the portfolio outcomes, this is modelled by a real-valued random variable $G$. 
The quantity $G$ depends on economic factors as follows: 
let us denote by $F$ a fixed deterministic functional  $F:C\left(\left[0,T\right];\mathbb{R}^{d}\right)\rightarrow\mathbb{R}_+$ which 
is $\mathscr{N}_T$-measurable. As the probability space is not fixed, for each $\pi\in\Pi$ we define
the corresponding reference point by $G^{\pi}:=F\left(Y^{\pi}_{\cdot}\right)$. That is, we assume that the benchmark is a non-negative functional of the
economic factors. Results can easily be extended to the slightly more general case where
$G^{\pi}:=F(Y^{pi}_{\cdot},B^{\pi}_{\cdot})$ for some functional $F$.

For any strategy $\pi\in\Pi$, we define the functionals 
\begin{equation}
V_{+}\left(\pi\right):=\int_{0}^{\infty}w_{+}\left(\mathbb{P^{\pi}}\left(u_{+}\left(\left(X_{T}^{\pi}-G^{\pi}\right)_{+}\right)>t\right)\right)dt,\label{behavioural functional 3}
\end{equation}
and 
\begin{equation}
V_{-}\left(\pi\right):=\int_{0}^{\infty}w_{-}\left(\mathbb{P^{\pi}}\left(u_{-}\left(\left(X_{T}^{\pi}-G^{\pi}\right)_{-}\right)>t\right)\right)dt.\label{eq:behavioural functional 4}
\end{equation}
The optimal portfolio problem for an investor under CPT consists in maximising
the following performance functional: 
\begin{equation} V\left(\pi\right):=V_{+}\left(\pi\right)-V_{-}\left(\pi\right),\label{eq:behavioural functional} \end{equation}
which is defined provided that at least one of the summands is finite.  
Fix $x>0$, $y\in\mathbb{R}$. Set $\Pi':=\{\pi\in\Pi(x,y):V_-(\pi)<\infty\}$
and define
\begin{equation}
V:=\sup_{\pi\in\Pi'}V\left(\pi\right). \label{eq:V funt}
\end{equation}
The value $V$ represents the maximal satisfaction achievable by investing in the stock and riskless asset in a CPT framework.
Our purpose is to prove the existence of $\hat{\pi}\in\Pi'$ such that $V(\hat{\pi})=V$.


\subsection{Main result}\label{mrr}

We make the following assumptions. Recall the notation $y_{t}^{\star}=\sup_{s\leqslant t} \left|y_{s}\right|$.

\begin{assumption} \label{acoef1}
The functionals $\kappa,\,\, \lambda,\,\, \theta\,\, \mbox{and}\,\, \nu$ are uniformly bounded on 
$\left[0,T\right]\times C\left(\left[0,T\right];\mathbb{R}^{d}\right)$.
Furthermore, for fixed $t\geqslant 0$ and functions  $y^{n},z\in C\left(\left[0,T\right];\mathbb{R}^{d}\right)$ 
such that $\left(y^{n}-z\right)^{\star}_t \rightarrow 0$, $n\to\infty$ we have
$\kappa_t\left(y^{n}_{\cdot}\right)\rightarrow \kappa_t\left(z_{\cdot}\right)$ and the same holds for 
the functionals 
$\lambda,\theta$ and $\nu$. We will refer to this as the coefficients being path-continuous at any time 
$t\in\left[0,T\right]$. 
\end{assumption}

\begin{assumption} \label{ul}
A (weak) solution of equation (\ref{eq1}) exists and it is unique in law.   
\end{assumption}

\begin{assumption}\label{u} 
We assume that $u_{\pm}:\mathbb{R}_{+}\rightarrow\mathbb{R}_{+}$ and 
$w_{\pm}:\left[0,1\right]\rightarrow\left[0,1\right]$ are continuous, non-decreasing functions with $u_{\pm}\left(0\right)=0,\, w_{\pm}\left(0\right)=0,\, w_{\pm}\left(1\right)=1,$ and  
\begin{equation} u_{+}\left(x\right)\leqslant k_{+}\left(x^{\alpha}+1\right), \label{u+} \mbox{for all } x\in\mathbb{R}_{+},\end{equation}  
\begin{equation} w_{+}\left(p\right)\leqslant g_{+}p^{\gamma}, \label{w+} \mbox{ for all } p\in\left[0,1\right],\end{equation}  
with $\gamma$, $\alpha>0$, $k_{+},g_{+}>0$ fixed constants. 
\end{assumption}

We denote by $L^p(\Omega,\mathbb{P})$ the usual space of $p$-integrable random variables on a
probability space $(\Omega,\mathcal{F},\mathbb{P})$.

\begin{assumption}\label{ghyp} There is $\vartheta>0$ such that 
$\vartheta\gamma>1$ and $G^{\pi}\in L^{\vartheta\gamma}\left(\Omega,\mathbb{P}^{\pi}\right)$ for all $\pi\in\Pi$.
\end{assumption}

Note that, under Assumption \ref{ul}, the law of $G^{\pi}$ is independent of $\pi$ and hence Assumption 
\ref{ghyp} holds iff $G^{\pi}\in L^{\vartheta\gamma}\left(\Omega,\mathbb{P}^{\pi}\right)$
for one particular $\pi$.
 
In order to ensure that the functional $V$ and the optimisation problem in (\ref{eq:V funt}) are defined over a non-empty set, we introduce the 
following assumption on $u_{-}$, the distortion function $w_{-}$ and the reference point $G^{\pi}$. 

\begin{assumption}\label{vminus}
The functions $w_{-}, u_{-}$ are such that,
for all $\pi\in\Pi$, 
\begin{equation}
\int_0^{\infty} w_{-}\left(\mathbb{P}^{\pi}\left(u_{-}\left(G^{\pi}\right)>y\right)\right)dy<\infty. \label{vmenos}
\end{equation}
\end{assumption}

This assumption ensures that the set $\Pi'$ is not empty. Indeed, let $(\Omega,\mathcal{F},\left\{\mathcal{F}_t\right\}_{0\leqslant t \leqslant T},\mathbb{P})$
be a filtered probability space where \eqref{ekuY} has a solution $Y_t$. Then setting $\phi_t:=0$ and $X_t:=x$ for all $t$,
$$
\left(\Omega,\mathcal{F},\left\{\mathcal{F}_t\right\}_{0\leqslant t \leqslant T},\mathbb{P},x,Y_t,
0,\left(B_t,W_t\right),(x,y)\right)
$$
belongs to $\Pi'$. Fix $x>0$ and $y\in\mathbb{R}$. Our main result can now be stated.

\begin{theorem} \label{wp1}
Under Assumptions \ref{positv}, \ref{acoef1}, \ref{ul}, \ref{u}, \ref{ghyp} and \ref{vminus} the problem 
\eqref{eq:V funt} is well-posed, i.e. $V<\infty$.
Moreover, there exists an optimal strategy $\hat{\pi}\in\Pi'$ attaining the supremum in
(\ref{eq:V funt}), i.e. $V=V\left(\hat{\pi}\right)$. 
\end{theorem}

\subsection{A relaxation of the set of controls}\label{ketto}

We introduce a relaxation of the problem by extending the class of investment strategies given in 
Definition \ref{def1}, we shall call this extension the class of auxiliary controls.
This relaxation is introduced in order to ensure the closedness of the set of laws of the processes $\left(Y_{\cdot},X_{\cdot}\right)$. 

We follow the martingale problem formulation, thus we refer to $a_t$ and $b_t$  as the drift/diffusion
coefficients of the process $\left(Y_t,X_t\right)$, as they appear  
in the martingale problem formulation of equations \eqref{eq1} and \eqref{eq2}.
In order to use \cite{Kry01},
these coefficients must take values in a family of convex subsets of $S^{d+1}_{+}\times \mathbb{R}^{d+1}$ hence we 
shall consider a 'convex extension' of the set in which the coefficients in equations (\ref{eq1}) and (\ref{eq2}) take values.
 
\begin{definition} \label{def3}
Denote $\mathbb{A}=S^{d+1}_{+}\times\mathbb{R}^{d+1}$.
For any pair of continuous functions $(x_{\cdot},y_{\cdot})\in C([0,T];\mathbb{R}\times\mathbb{R}^{d})$ 
and for any $t\in\left[0,T\right]$ we define
\begin{eqnarray}\nonumber
A_t\left(x_{\cdot},y_{\cdot}\right) &=& \left\{(a,b)\in \mathbb{A} \middle| (a,b)=
\left(\left(\begin{array}{cc} \frac{1}{2}\kappa\kappa^{\star}\left(t,y_{\cdot}\right) & 0  \\ 0 & 
\frac{1}{2}m\lambda^2\left(t,y_{\cdot}\right)x_t^2 \\ \end{array} \right),
\left(\begin{array}{c} \nu\left(t,y_{\cdot}\right)\\ l\theta\left(t,y_{\cdot}\right)x_t \end{array} \right) \right),
\right.\\
\label{Atdef} & & \left.\begin{array}{c}
          0\leqslant m \leqslant 1,\\
         0\leqslant l \leqslant \sqrt{m} \\
  \end{array} \right\}
\end{eqnarray}
\end{definition}

\begin{remark}\label{ordiniaire} {\rm
Notice that, for any investment strategy $\pi$ as in Definition \ref{def1},  
if $\sigma_t=\left(\begin{array}{cc} \kappa\left(t,y_{\cdot}\right) & 0  \\
0 & \phi_t\lambda\left(t,y_{\cdot}\right)x_t \\ \end{array} \right)$ 
and $b_t=\left(\begin{array}{c} \nu\left(t,y_{\cdot}\right)    \\
\phi_t\theta\left(t,y_{\cdot}\right)x_t \\ \end{array}\right)$ then, 
defining $a_t=\frac{1}{2}\sigma_t\sigma_t^{\star}$, 
the pair $\left(a_t,b_t\right)$ belongs to $A_t\left(x_{\cdot},y_{\cdot}\right)$.}  
\end{remark}

The following definition describes the 
familiy of auxiliary controls used throughout this work. It stresses the fact of having It\^{o} 
processes whose coefficients belong to the convex sets $A_t\left(x_{\cdot},y_{\cdot}\right)$ in 'a measurable way' as $t,x_{\cdot}$ and $y_{\cdot}$ vary. 


\begin{definition}\label{bPi}
We define a familiy of auxiliary controls $\overline{\Pi}=\overline{\Pi}(x,y)$. Namely, an auxiliary control 
$\overline{\pi}\in\overline{\Pi}$ consists of a collection 
$$
\overline{\pi}:=\left(\Omega,\mathcal{F},\left\{\mathcal{F}_t\right\}_{0\leqslant t\leqslant T},\mathbb{P},X_t,Y_t,\left(B_t,W_t\right),(x,y)\right)
$$ 
where $x>0$, $y\in\mathbb{R}$,
\begin{description}
\item[(a)] $\left(\Omega,\mathcal{F},\left\{\mathcal{F}_t\right\}_{t\geqslant 0},\mathbb{P}\right)$ is a complete filtered probability space whose filtration satisfies the usual conditions;
\item[(b)] $\xi_t:=\left(B_t,W_t\right)$ is an $\mathbb{R}^{d+1}$-valued standard $\mathcal{F}_t$-Brownian motion;
\item[(c)] there exists an $\mathbb{A}$-valued, $ \mathcal{F}\otimes\mathcal{B}\left(\left[0,T\right]\right)$-measurable and $\mathcal{F}_t$-adapted process, 
denoted by $\left(a_t,b_t\right)$, such that (d) and (e) below hold;  
\item[(d)] $X_t$ and $Y_t$ are $\mathcal{F}\otimes\mathcal{B}\left(\left[0,T\right]\right)$-measurable and 
$\mathcal{F}_t$-adapted such that a.s. for all $t\geqslant 0$;
\begin{equation} \label{atbt}
\left(\begin{array}{c} Y_t  \\  X_t \\ \end{array}\right)=\left(\begin{array}{c} y  \\  x \\ 
\end{array}\right) + \int_0^t \sqrt{2 a_s}d\xi_s+ \int_0^t b_s ds;
\end{equation}
\item[(e)] for almost all $\left(\omega,t\right)\in\Omega\times\left[0,T\right]$, we have $\left(a_t,b_t\right)\in A_t\left(X_{\cdot},Y_{\cdot}\right)$. 
\end{description}
\end{definition}

We will often write $X^{\overline{\pi}}$, $Y^{\overline{\pi}}$ to indicate that we mean $X$, $Y$ belonging
to $\overline{\pi}$.

For each $\overline{\pi}\in\overline{\Pi}$, we can define $V_{\pm}(\overline{\pi})$ as before and
we can set $V(\overline{\pi}):=V_+(\overline{\pi})-V_-(\overline{\pi})$ for $\overline{\pi}\in\overline{\Pi}':=
\{\pi\in\overline{\Pi}:V_-({\pi})<\infty\}$.

\begin{remark}\label{landm}
For a pair of processes processes $a_t$ and $b_t$ in 
$A_t\left(X_{\cdot},Y_{\cdot}\right)$ one can define the corresponding real-valued processes $l_t$ and $m_t$ with 
$0\leqslant m_t \leqslant 1, 0\leqslant l_t \leqslant \sqrt{m_t}$ setting
$$
l_t:=b^{d+1}_t1_{\theta_t(t,Y_t)\neq 0}/(X_t\theta_t(t,Y_t)),\ m_t:=
a^{d+1,d+1}1_{\lambda_t(t,Y_t)\neq 0}/(X_t^2\lambda^2(t,Y_t)).
$$
Conditions \normalfont{(c),(d)} in Definition \ref{bPi} together with Assumption \ref{acoef1} imply
that 
$l_t,m_t$ can be chosen
$\mathcal{F}\otimes\mathcal{B}\left(\left[0,T\right]\right)$ measurable and $\mathcal{F}_t$-adapted.
\end{remark}  

Equation (\ref{atbt}) can be rewritten as the set of equations below. Denote 
\[
\sigma_t=\left(\begin{array}{cc} \kappa\left(t,Y_{\cdot}\right) & 0  \\ 0 & \sqrt{m_t}\lambda\left(t,Y_{\cdot}\right)X_t \\ \end{array} \right)
\]
and 
\[
b_t=\left(\begin{array}{c} \nu\left(t,Y_{\cdot}\right)    
\\  l_t\theta\left(t,Y_{\cdot}\right)X_t \\ \end{array}\right).
\]
Setting $a_t:=\frac{1}{2}\sigma_t\sigma^{\star}_t$,  
\begin{equation}\label{eq3}
Y_t=y+\int_0^t\nu_s\left(Y_{\cdot}\right)ds+\int_0^t\kappa_s\left(Y_{\cdot}\right)dB_s, 
\end{equation}
\begin{equation}\label{eq4}
X_t=x+\int_0^t l_s\theta_s\left(Y_{\cdot}\right)X_sds+\int_0^t \sqrt{m_s}\lambda_s\left(Y_{\cdot}\right)X_sdW_s. 
\end{equation}


\begin{definition}\label{prtf}
Let $\overline{\pi}\in\overline{\Pi}$ be a relaxed control. We say that $X^{\overline{\pi}}_t$ is a 
\textsl{portfolio value} process if $l_t=\sqrt{m_t}$, i.e. 
\begin{equation}
dX_t=\sqrt{m_t}\theta\left(t,Y_{\cdot}\right)X_tdt+\sqrt{m_t}\lambda\left(t,Y_{\cdot}\right)X_tdW_t.
\end{equation}
\end{definition}

\begin{remark}\label{threestars}
If $X^{\overline{\pi}}_t$ is a 
{portfolio value} process then, taking $\phi_t=\sqrt{m_t}$,
we can see that 
$$
\left(\Omega^{\overline{\pi}},\mathcal{F}^{\overline{\pi}},
\left\{\mathcal{F}^{\overline{\pi}}_t\right\}_{t\geqslant 0},
\mathbb{P}^{\overline{\pi}},X^{\overline{\pi}}_t,
Y^{\overline{\pi}}_t,\phi_t,\left(B^{\overline{\pi}}_t,W^{\overline{\pi}}_t\right),(x,y)\right)
$$
belongs to $\Pi$.
\end{remark}

\begin{remark} \label{exprp}
Suppose that we are given a $\overline{\pi}\in\overline{\Pi}$ i.e. there is a standard $d+1$-dimensional 
Brownian motion $\left(B,W\right)$ on 
$\left(\Omega^{\overline{\pi}},\mathcal{F}^{\overline{\pi}},\left\{\mathcal{F}^{\overline{\pi}}_t\right\}_{t\geqslant 0},\mathbb{P}^{\overline{\pi}}\right)$ and processes 
$X^{\overline{\pi}}_t, Y^{\overline{\pi}}_t, m^{\overline{\pi}}_t,l^{\overline{\pi}}_t$ such that equations (\ref{eq3}) and (\ref{eq4}) hold.
Define the continuous semimartingale 
$M^{\overline{\pi}}_t:=
\int_0^t \sqrt{m_s^{\overline{\pi}_t}}\lambda_s\left(Y_{\cdot}\right)dW_s+
\int_0^t l_s^{\overline{\pi}_t}\theta_s\left(Y_{\cdot}\right)ds$. Then we can rewrite equation (\ref{eq4}) as
\begin{equation}\label{abrafaxe}
X_t = x + \int_0^t X_sdM^{\overline{\pi}}_s.
\end{equation}
Equation \eqref{abrafaxe} has a unique strong solution on the given
probability space, given the stochastic exponential 
\begin{equation}
X^{\overline{\pi}}_t= x\exp\left\{\int_0^t\sqrt{m_s}\lambda_s\left(Y^{\overline{\pi}}_{\cdot}\right)
dW^{\overline{\pi}}_s+\int_0^t \left[l_s\theta_s\left(Y^{\overline{\pi}}_{\cdot}\right)-
\frac{1}{2}m_s \lambda_s^2\left(Y^{\overline{\pi}}_{\cdot}\right)\right]ds\right\},
\end{equation}
and this process is positive $\mathbb{P}^{\overline{\pi}}$-a.s.  
\end{remark}

\subsection{Krylov's theorem and related results}

\begin{lemma}\label{rcoef}
Let $M=\max\left\{\| \kappa\|_{\infty},\|\lambda\|_{\infty},\|\theta\|_{\infty},\|\nu\|_{\infty} \right\}$. 
The set $A_t\left(x_{\cdot},y_{\cdot}\right)$ is convex, closed and bounded, where the bound depends on $M$ and $x_t$ only.
\end{lemma}
\begin{proof}
Notation $|\cdot|$ will refer to Euclidean norms of varying dimensions. For simplicity, we assume $d=1$.
Notice that 
\[
|\left(\sigma_{t},b_{t}\right)|=\left(\kappa_{t}^{2}\left(y_{\cdot}\right)+m^{2}\lambda_{t}^{2}
\left(y_{\cdot}\right)x_{t}^{2}+\nu_{t}^{2}\left(y_{\cdot}\right)+l^{2}\theta_{t}^{2}\left(y_{\cdot}\right)
x_{t}^{2}\right)^{1/2},
\]
hence 
\begin{equation}
|\left(\sigma_{t},b_{t}\right)|\leqslant\sqrt{2}M\left(1+\left|x_{t}\right|\right). \label{bndeq}
\end{equation}


It is clear that the set is closed. For a fixed $t$, $x_{\cdot}$ and $y_{\cdot}$ the set is bounded. 
Indeed, let $\left(a_{t},b_{t}\right)\in A_{t}\left(x_{\cdot},y_{\cdot}\right)$. Then 
we have 
\[
|\left(a_{t},b_{t}\right)|=\left(\frac{1}{4}\cdot\left(\kappa_{t}^{2}\left(y_{\cdot}\right)\right)^2+
\frac{1}{4}\cdot \left(m\cdot\lambda_{t}^{2}\left(y_{\cdot}\right)\cdot x_{t}^{2}\right)^2+\nu_{t}^{2}
\left(y_{\cdot}\right)+l^{2}\theta_{t}^{2}\left(y_{\cdot}\right)x_{t}^{2}\right)^{1/2},
\]
so
\[
|\left(a_{t},b_{t}\right)|\leqslant\left(\frac{1}{4}M^4+\frac{1}{4}\cdot M^4\cdot\left( x_{t}^{2}\right)^{2}+
M^2+M^{2}x_{t}^{2}\right)^{1/2},
\]
which leads to $|\left(a_{t},b_{t}\right)|\leqslant\frac{1}{2}\left(M+1\right)^2+M^2x^{2}_t$.

In particular, \begin{equation}\label{star}
\|A_{t}\left(x_{\cdot},y_{\cdot}\right)\|:=
\max\left\{ \left| \left(a_{t},b_{t}\right)\right| :
\left(a_{t},b_{t}\right)\in A_{t}\left(x_{\cdot},y_{\cdot}\right)\right\} \leqslant K\left(1+\left|x_{t}\right|^2
\right),
\end{equation}
for some $K\geq 0$.

The set $A_t\left(x_{\cdot},y_{\cdot}\right)$ is also convex.
Indeed, let $\left(\alpha,b\right),\left(\gamma,c\right)\in A_{t}\left(x_{\cdot},y_{\cdot}\right)$
then, for $0\leq \mu\leq 1$, 
\begin{gather*}
\mu\left(\alpha,b\right)+(1-\mu)\left(\gamma,c\right)=\\
\left(\left(\begin{array}{cc}
\frac{1}{2}\kappa_{t}^{2}\left(y_{\cdot}\right) & 0\\
0 & \frac{1}{2}\left(\mu m+(1-\mu)n\right)\lambda_{t}^{2}\left(y_{\cdot}\right)x_{t}^{2}
\end{array}\right),\left(\begin{array}{c}
\nu_{t}\left(y_{\cdot}\right)\\
\left(\mu l+(1-\mu)p\right)\theta_{t}\left(y_{\cdot}\right)x_{t}
\end{array}\right)\right)
\end{gather*}
 with $0\leqslant m,n\leqslant 1$, $0\leqslant l\leqslant\sqrt{m}$ and $0\leqslant p\leqslant\sqrt{n}$. Clearly,
$\mu l+(1-\mu)p\leqslant\sqrt{\mu m+\left(1-\mu\right)n}$, by concavity of the square root function.
\end{proof}

 In order to deal with (semi)continuity issues related to the family of sets defined in Definitions \ref{def3} and 
 \eqref{Atdef}, the support functions of sets $A_t\left(x_{\cdot},y_{\cdot}\right)$ are now considered.
We denote for all $u\in\mathbb{R}^{(d+1)(d+1)}$, $v\in\mathbb{R}^{d+1}$ and $t\in\left[0,T\right]$,
\begin{equation}
F_{t}\left(x_{\cdot},y_{\cdot}\right)\left(u,v\right)=\max\left\{ \sum_{i,j}a_{ij}u_{ij}+\sum_{j}b_{j}v_{j}\,:\,\left(a,b\right)\in A_{t}\left(x_{\cdot},y_{\cdot}\right)\right\}.
\end{equation}

Under Assumption \ref{acoef1}, for fixed $t\geqslant 0$ and $\left(u,v\right)$,
the support function $(x_{\cdot},y_{\cdot})\to F_{t}\left(x_{\cdot},y_{\cdot}\right)\left(u,v\right)$ is 
continuous, since we are fixing $t$, restricting the trajectories to $[0,t]$, and thus the 
$\max$ is taken over a compact set by Lemma \ref{rcoef}.
In particular, the set $A_t\left(x_{\cdot},y_{\cdot}\right)$ is upper-semicontinuous in the sense 
of Assumption 3.1 iii) in \cite{Kry01}. It is also clear that, 
for fixed $u,v\in \mathbb{A}$, $F_t\left(u,v,x_{\cdot},y_{\cdot}\right)$ is a Borel function 
on $\left[0,T\right]\times C\left(\left[0,T\right];\mathbb{R}^{d+1}\right)$.

We now present some moment estimates which will, in particular, guarantee tightness for the family
of the laws of $(X^{\overline{\pi}},Y^{\overline{\pi}})$, $\overline{\pi}\in\overline{\Pi}$ in $C\left(\left[0,T\right];\mathbb{R}^{d+1}\right)$.

\begin{proposition} \label{mestim1}
For the ease of reference we denote $\zeta_{t}=\left(Y_t,X_t\right)$. Under Assumption \ref{acoef1}, for any $m>0$, 
\begin{equation} \label{totbnd}
\sup_{\overline{\pi}\in\overline{\Pi}}\mathbf{E}_{\overline{\pi}}\left[\sup_{t\leqslant T}| \zeta^{\overline{\pi}}_t|^m\right]<\infty.
\end{equation}
\end{proposition}

\begin{proposition}\label{tgh}
Under Assumption \ref{acoef1}, let $\overline{\pi}\in\overline{\Pi}$ and $\left(Y_{t}^{\overline{\pi}},X_{t}^{\overline{\pi}}\right)$ its associated processes solving (\ref{eq3}) and (\ref{eq4}).
Then, there exists a constant $K>0$ not depending on $\overline{\pi}\in\overline{\Pi}$, such that for any 
$\eta>0$ and $s,t\in\left[0,T\right]$,   
\begin{equation}
\mathbf{E}_{\overline{\pi}}\|\zeta_{t}-\zeta_{s}\|^{\eta}\leqslant K\,|t-s|^{\frac{\eta}{2}}.\label{eq:tightnessK}
\end{equation}
 \end{proposition}

See the Appendix for a standard proof of both propositions above.
A well-known result on tightness of measures on 
$C\left([0,T];\mathbb{R}^{d+1}\right)$ gives the following corollary. This could also be obtained
by the method of Theorem 3.2 in \cite{krli}.

\begin{corollary}\label{relcomp}
Let Assumption \ref{acoef1} be in force.
Let $\left\{ \overline{\pi}_{n}\right\} \subset\overline{\Pi}$. The set of laws of the process 
$\zeta_{\cdot}^{\overline{\pi}_{n}}$ on $C\left(\left[0,T\right];\mathbb{R}^{d+1}\right)$ is 
relatively weakly compact. \hfill $\Box$
\end{corollary}

Now we restate Theorem 3.2 of \cite{Kry01} in our setting, which will provide weak compactness
of the distributions of weak controls.

\begin{theorem} \label{thm3.2} Let Assupmtion \ref{acoef1} be in force.
Denote by $\mathbb{Q}^{\overline{\pi}}$ the distribution of $\zeta^{\overline{\pi}}_{\cdot}$ on 
$C\left(\left[0,T\right];\mathbb{R}^{d+1}\right)$. Then the set $\left\{\mathbb{Q}^{\overline{\pi}} \,:\,\overline{\pi}\in\overline{\Pi}\right\}$ is 
sequentially weakly compact: for any sequence 
$\overline{\pi}_n\in\overline{\Pi}$ there is a subsequence $n(m)\rightarrow\infty$ as $m\rightarrow\infty$ and a
$\overline{\pi}\in\overline{\Pi}$ such that for any real-valued, bounded, continuous 
function $H(x_\cdot)$ on $C\left(\left[0,T\right];\mathbb{R}^{d+1}\right)$ we have 
\begin{equation}\label{eqt32}
\lim_{m\rightarrow\infty} E^{\nu_m} H\left(\zeta^{\nu_m}_{\cdot}\right)=E^{\pi}H\left(\zeta^{\pi}_{\cdot}\right),
\end{equation}
where $\nu_m=\pi_{n(m)}$.
\end{theorem}
\begin{proof} It follows from the above discussions that Assumption 3.1 ii) and iii) in \cite{Kry01}
hold in the present case. One does not have Assumption 3.1 i) of \cite{Kry01} though
(linear growth condition on $\|A_t\left(X_{\cdot},Y_{\cdot}\right)\|$),
there is a quadratic growth instead, see \eqref{star}. But, as Corollary
\ref{relcomp} shows, this is still sufficient to get tightness (and hence relative weak compactness)
of the sequence $\mathbb{Q}^{\overline{\pi}_n}$ in our setting. Then one can check that the proof of Theorem 3.2
in \cite{Kry01} goes through and we can conclude.
\end{proof}

The next lemma shows that, to any auxiliary control $\overline{\pi}$ in the sense of Definition \ref{bPi},
we can associate an \textit{investment strategy} (in the sense of Definition \ref{def1}) 
with higher value function.

\begin{lemma}\label{pistar}
Let $$
\overline{\pi}=\left(\Omega^{\overline{\pi}},
\mathcal{F}^{\overline{\pi}},\left\{\mathcal{F}^{\overline{\pi}}_t\right\}_{t\geqslant0},
\mathbb{P}^{\overline{\pi}},\left(X^{\overline{\pi}},
Y^{\overline{\pi}}\right),\left(B^{\overline{\pi}},W^{\overline{\pi}}\right),(x,y)\right)\in\overline{\Pi}.
$$
Then a solution to  
\begin{equation}
dY_t=\nu_t\left(Y_{\cdot}\right)dt+\kappa_t\left(Y_{\cdot}\right)dB_t,\quad Y_0=y, 
\end{equation}
\begin{equation}\label{cheney}
d\hat{X}_t=\sqrt{m_t}\theta_t\left(Y_{\cdot}\right)\hat{X}_tdt+\sqrt{m_t}\lambda_t\left(Y_{\cdot}\right)\hat{X}_tdW_t,
\quad X_0=x,
\end{equation}
exists on the same filtered probability space and $\hat{X}_T\geq X_T^{\overline{\pi}}$
a.s. Furthermore, $\hat{X}_t$ is a portfolio value process.
\end{lemma}
\begin{proof}
Let us define  
\[
Z_t:=\exp\left(-\int_0^t \left(l_s^{\overline{\pi}}-\sqrt{m^{\overline{\pi}}_s}\right)\left(X^{\overline{\pi}}_{\cdot},Y^{\overline{\pi}}_{\cdot}\right)\theta\left(Y^{\overline{\pi}}_{\cdot}\right)ds \right)
\]
and set $\hat{X}_t:=Z_tX_t^{\overline{\pi}}$. 
It\^{o}'s formula shows that $\hat{X}_t$ indeed verifies \eqref{cheney}.
Since $\theta_t\geq 0$ was assumed, we get that $Z_t\geq 1$ hence $\hat{X}_t\geq X_t^{\overline{\pi}}$, for all $t$.
\end{proof}

\subsection{Proof of Theorem \ref{wp1}}

\begin{proof}
Let $t>0$. By (\ref{w+}) and (\ref{u+}),
\[
w_{+}\left(\mathbb{Q}^{\pi}\left(u_{+}\left(\left(X_{T}^{\pi}-G^{\pi}\right)_{+}\right)>t
\right)\right)\leqslant g_{+}\left[\mathbb{Q}^{\pi}\left(\left(X_{T}^{\pi}-G^{\pi}\right)_{+}^{\alpha}>\frac{t}{k_{+}}-1\right)\right]^{\gamma}.
\]

Hence, 
\begin{eqnarray*}
V_{+}\left(\pi\right) & \leqslant g_{+}\int_{0}^{\infty}\left[\mathbb{Q}^{\pi}\left(\left(X_{T}^{\pi}-G^{\pi}\right)_{+}^{\alpha}>\frac{t}{k_{+}}-1\right)\right]^{\gamma}=&\\
\,\,\,\,\,\, & =g_{+}\left(1+\int_{k_{+}}^{\infty}\left[\mathbb{Q}^{\pi}\left(\left(X_{T}^{\pi}-G^{\pi}\right)_{+}^{\alpha}>\frac{t}{k_{+}}-1\right)\right]^{\gamma}\right),&
\end{eqnarray*}
\begin{equation}
\int_{k_{+}}^{\infty}\left[\mathbb{Q}^{\pi}\left(\left(X_{T}^{\pi}-
G^{\pi}\right)_{+}^{\alpha}>\frac{t}{k_{+}}-1\right)\right]^{\gamma}dy\leqslant k_{+}\int_{0}^{\infty}
\left[\mathbb{Q}^{\pi}\left(\left(X_{T}^{\pi}-G^{\pi}\right)_{+}^{\alpha}>s\right)\right]^{\gamma}dx.
\end{equation}

If $s\geq 1$, applying Chebyshev's inequality and Assumption \ref{ghyp},
\begin{equation}\label{integrable}
\left[\mathbb{Q}^{\pi}\left(\left(X_{T}^{\pi}-G^{\pi}\right)_{+}^{\alpha}>s\right)\right]^{\gamma}=
\left[\mathbb{Q}^{\pi}\left(\left(X_{T}^{\pi}-G^{\pi}\right)_{+}^{\alpha\vartheta}>s^{\vartheta}\right)\right]^{\gamma}\leqslant\frac{\left[\mathbf{E}_{\pi}\left(X_{T}^{\pi}-G^{\pi}\right)^{\alpha\vartheta}_{+}\right]^{\gamma}}{s^{\vartheta\gamma}}\leqslant M^{\gamma}\frac{1}{s^{\vartheta\gamma}},
\end{equation}
where $M=\sup_{\pi}\mathbf{E}_{\pi}\left(X_{T}^{\pi}\right)^{\alpha\vartheta}_{+}<\infty$ (note that $G\geqslant 0$),
by Proposition \ref{mestim1}. Note that $1/s^{\vartheta\gamma}$ is integrable on $[1,\infty)$.  

Hence the problem is well-posed since $V(\pi)\leq V_+(\pi)$ for all $\pi\in\Pi'$ and we have just seen that
the latter has an upper bound independent of $\pi$.

By Theorem \ref{thm3.2} the set of 
laws $\left\{ \mathbb{Q}^{\pi}\right\}$, $\pi\in\overline{\Pi}$ of the processes 
$\zeta_{\cdot}^{\pi}=\left(X_{\cdot}^{\pi},Y_{\cdot}^{\pi}\right)$ is relatively compact in the weak 
topology. Let $\left\{ \pi^{n}\right\} $ be sequence of weak controls $\pi^{n}\in\overline{\Pi}'$ such that 
\begin{equation}
V\left(\pi^{n}\right)\rightarrow\sup_{\pi\in\overline{\Pi}'}V\left(\pi\right),\,\,\,\, n\rightarrow\infty.
\end{equation}

There is a subsequence of $\left\{ \pi^{n}\right\} $ denoted by $\left\{ \pi^{k}\right\} $ such that 
$\mathbb{Q}^{\pi_{k}}\Rightarrow\mathbb{Q}^{\pi^{\star}}$
as $k\rightarrow\infty$ and $\pi^{\star}\in\overline{\Pi}$. 

By Skorokhod's theorem there is a probability space, that will be denoted by $\left(\tilde{\Omega},\tilde{\mathcal{F}},\tilde{\mathbb{P}}\right)$
and random variables $\tilde{X}_{\cdot}^{k},\tilde{Y}_{\cdot}^{k}:
\left(\tilde{\Omega},\tilde{\mathcal{F}},\tilde{\mathbb{P}}\right)\rightarrow$ ${C}
\left(\left[0,T\right];\mathbb{R}\right)$, ${C}\left(\left[0,T\right];\mathbb{R}^{d}\right)$, respectively, such that the law of 
$\left(\tilde{X}_{\cdot}^{k},\tilde{Y}_{\cdot}^{k}\right)$ equals $\mathbb{Q}^{\pi_{k}}$ 
and $\tilde{X},\tilde{Y}:\left(\tilde{\Omega},\tilde{\mathcal{F}},\tilde{\mathbb{P}}\right)$
$\rightarrow{C}\left(\left[0,T\right];\mathbb{R}\right)$, 
${C}\left(\left[0,T\right];\mathbb{R}^{d}\right)$ with law equal to $\mathbb{Q}^{\pi^{\star}}$ 
such that $\tilde{X}^{k}\rightarrow\tilde{X}$, $\tilde{Y}^{k}\rightarrow\tilde{Y}$ a.s. in the uniform norm.

By Assumption \ref{ul}, $\tilde{Y}^k$ and $\tilde{Y}$ have the same law and 
$\tilde{Y}^k\rightarrow\tilde{Y}$ in probability (even a.s.). 
By Th\'eor\`eme 1 in \cite{BEKSY} 
$F\left(\tilde{Y}^k\right)\rightarrow F\left(\tilde{Y}\right)$ in probability.

By continuity of $u_{\pm}$ and the projection $p_T(f):=f(T)$, $f\in {C}
\left(\left[0,T\right];\mathbb{R}\right)$, we also have 
$u_{\pm}\left(\left(\tilde{X}_T^k-F\left(\tilde{Y}^k\right)\right)_{\pm}\right)\rightarrow$ 
$u_{\pm}\left(\left(\tilde{X}_T-F\left(\tilde{Y}\right)\right)_{\pm}\right)$ in probability.

It follows that, denoting by $D$ the set of discontinuity points of the cumulative distribution
functions of $u_{\pm}\left(\left(\tilde{X}_T-F\left(\tilde{Y}\right)\right)_{\pm}\right)$, 
for any $y\in\mathbb{R}\setminus D$ we have
\[
\mathbb{Q}^{\pi_{k}}\left(u_{\pm}\left(\left(X_{T}^{\pi_{k}}-G^{\pi_{k}}\right)_{\pm}\right)>y\right)\rightarrow\mathbb{Q}^{\pi^{\star}}\left(u_{\pm}\left(\left(X_{T}^{\pi^{\star}}-G^{\pi^{\star}}\right)_{\pm}\right)>y\right)
\]
as $k\rightarrow\infty$.

Since $w_{\pm}$ are continuous, also 
\[w_{\pm}\left(\mathbb{Q}^{\pi_{k}}\left(u_{\pm}\left(\left(X_{T}^{\pi_{k}}-G^{\pi_{k}}\right)_{\pm}\right)>y\right)\right)\rightarrow w_{\pm}\left(\mathbb{Q}^{\pi^{\star}}\left(u_{\pm}\left(\left(X_{T}^{\pi^{\star}}-G^{\pi^{\star}}\right)_{\pm}\right)>y\right)\right),\]
for $y\notin D$.
By Fatou's lemma,
\begin{eqnarray*}
\int_{0}^{\infty}w_{-}\left(\mathbb{Q}^{\pi^{\star}}\left(u_{-}\left(\left(X_{T}^{\pi^{\star}}-G^{\pi^{\star}}\right)_{-}\right)>y\right)\right)dy\leqslant\\
\underline{\lim}_{k}\int_{0}^{\infty}w_{-}\left(\mathbb{Q}^{\pi_{k}}\left(u_{-}\left(\left(X_{T}^{\pi_{k}}-G^{\pi_{k}}\right)_{-}\right)>y\right)\right)dy,\label{eq:fatou1}
\end{eqnarray*}
and, by \eqref{integrable} and the Fatou lemma,
\begin{eqnarray*}
\int_{0}^{\infty}w_{+}\left(\mathbb{Q}^{\pi^{\star}}\left(u_{+}\left(\left(X_{T}^{\pi^{\star}}-G^{\pi^{\star}}\right)_{+}\right)>y\right)\right)dy\geqslant\\
\overline{\lim}_{k}\int_{0}^{\infty}w_{+}\left(\mathbb{Q}^{\pi_{k}}\left(u_{+}\left(\left(X_{T}^{\pi_{k}}-G^{\pi_{k}}\right)_{+}\right)>y\right)\right)dy,
\end{eqnarray*}

It follows that $V(\pi^{\star})=\sup_{\pi\in\overline{\Pi}'}V(\pi)$. It is also clear that $\pi^{\star}\in
\overline{\Pi}'$.
{Let} $\left(a_{t},b_{t}\right)$ be the $\mathbb{A}$- valued processes associated to $\pi^{\star}\in\overline{\Pi}$ 
as in Definition \ref{bPi}. 
By Lemma \ref{pistar} there is 
$$
\pi'=\left(\Omega^{\pi^{\star}},\mathcal{F}^{\pi^{\star}},
(\mathcal{F}^{\pi^{\star}}_t)_{0\leqslant t\leqslant T},\mathbb{P}^{\pi^{\star}},X^{\pi'},Y^{\pi^{\star}},
\left(B^{\pi^{\star}},W^{\pi^{\star}}\right),(x,y)\right)
$$ 
which is a portfolio value process in the sense of Definition \ref{prtf}
and for which
\[
u_{+}\left(\left(X_{T}^{\pi'}-G^{\pi'}\right)_{+}\right)\geqslant 
u_{+}\left(\left(X_{T}^{\pi^{\star}}-G^{\pi^{\star}}\right)_{+}\right),
\]
notice that $G^{\pi'}=G^{\pi^{\star}}$ and 
$u_{-}\left(\left(X_{T}^{\pi^{\star}}-G^{\pi^{\star}}\right)_{-}\right)\geqslant 
u_{-}\left(\left(X_{T}^{\pi'}-G^{\pi'}\right)_{-}\right)$ also. Hence
$V(\pi')\geq \sup_{\pi\in\overline{\Pi}'} V(\pi)$.
\normalsize{Thus}, recalling Remark \ref{threestars},
the investment strategy $$
\hat{\pi}=\left(\Omega^{\pi\star},\mathcal{F}^{\pi\star},\mathbb{P}^{\pi\star},
\left\{ \mathcal{F}_{t}^{\pi\star}\right\}_{0\leqslant t\leqslant T},
X_{\cdot}^{\pi'},Y_{\cdot}^{\pi\star},\sqrt{m^{\pi'}_{\cdot}},(B^{\pi\star},W^{\pi\star}),(x,y)\right)
$$ 
is optimal i.e. $\sup_{\pi\in\Pi'}V\left(\pi\right)\leq V(\pi^{\star})\leq V\left(\hat{\pi}\right)=V(\hat{\pi})$
and, obviously, $\hat{\pi}\in\Pi'$.
\end{proof}  

\section{Extensions}\label{ext}

Based on economic considerations, we extend the model that was developed in the last section, 
by allowing the portfolio value process to influence the factor modelled by 
$Y_t$, the influence being 'additive'. 
This may be an appropiate model for e.g. a large investor.
Furthermore, a riskless asset with deterministic interest rate $r_t$ at time $t$ is included.
For the sake of simplicity we will assume that the factor process $Y$ is one-dimensional, the
results can be extended to the multidimensional case in a trivial way.
\begin{definition} \label{defe}
Let $\nu\left(t,y_{\cdot}\right)$ be a $\mathbb{R}$-valued process, such that the restriction of  $\nu$ to 
$\left[0,t\right]\times C\left(\left[0,T\right];\mathbb{R}\right)$ is 
$\mathcal{B}\left([0,t]\right)\otimes \mathscr{N}_{t}$-measurable, for any $0\leqslant t\leqslant T$.

Similarly, we define the $\mathbb{R}$-valued 
coefficients $\theta, \lambda, \rho, \kappa $ to have the same measurability.
\end{definition}

In this case, the stochastic differential equations of the optimal investment model are given by
\begin{equation}
dY_{t}=\nu\left(t,Y_{\cdot}\right)dt+\kappa\left(t,Y_{\cdot}\right)dB_{t}+\rho\left(t,X_{\cdot}\right)dX_t,
\end{equation}
\begin{equation}
dX_{t}=\phi_t\theta\left(t,Y_{\cdot}\right)X_{t}dt+\phi_t\lambda\left(t,Y_{\cdot}\right)X_{t}dW_{t}+
\left(1-\phi_t\right)r_{t}X_{t}dt,
\end{equation}
where $\phi_t\in [0,1]$ represents the proportion of wealth invested in the stock, $Y$ is an economic
factor, $X$ is the value process of the given portfolio strategy $\phi$. The set $\Pi$ can
be defined analogouly to Definition \ref{def1}.

\begin{assumption}\label{inter}
For all $t\geqslant 0$  the growth rate of the stock is greater than the growth rate of the bond, i.e.  
for all $t,y_{\cdot}$,
\begin{equation}
\theta\left(t,y_{\cdot}\right)\geqslant r_t\geqslant 0, \,\,\,\mathbb{P}^\pi-\mbox{a.s.}
\end{equation}
The functionals $\nu,\theta,\lambda,\kappa, \rho$ are bounded and path-continuous in the sense of 
Assumption \ref{acoef1}.
\end{assumption}

\begin{assumption}\label{ghyp2}
The reference point $G$ is a constant.
\end{assumption}
As in Subsection \ref{ketto}, we consider a relaxed setting. With this purpose in mind, we define 
$\theta^r\left(t,y_{\cdot}\right)=\theta\left(t,y_{\cdot}\right)-r_t$.
In what follows, $E$ is the $2\times 2$ matrix such that $E^{11}=1$ and $E^{ij}=0$ otherwise.
\begin{definition} \label{defAb}
We define the following family of sets. 
\begin{equation}\label{A2tdef}
A_t\left(x_{\cdot},y_{\cdot}\right) = 
\left\{(a,b)\in\mathbb{A} \left| a=\frac{1}{2}\kappa^2
\left(t,y_{\cdot}\right)E+\frac{1}{2}m
\lambda^2\left(t,y_{\cdot}\right)x^{2}_t\left(\begin{array}{cc} \rho^2\left(t,x_{\cdot}\right) & \rho\left(t,x_{\cdot}\right) \\ \rho\left(t,x_{\cdot}\right) & 1\end{array}\right)\right. ,\right.
\end{equation}
\begin{equation}
\left. \,  b=\left(\begin{array}{c}
         \nu\left(t,y_{\cdot}\right) \\ 0 
  \end{array} \right) 
  + \left(l x_t\theta^r\left(t,y_{\cdot}\right)+r_t x_t\right)\left(\begin{array}{c} \rho\left(t,x_{\cdot}\right) \\ 1 \end{array} \right),\,\, 
   \begin{array}{c}
         0\leqslant m\leqslant 1 \\
          0\leqslant l\leqslant \sqrt{m}
  \end{array} 
\right\}
\end{equation}
\end{definition}

The following lemma is crucial: it enables us to use results of \cite{Kry01}.

\begin{lemma} \label{A2}
The set $A_{t}\left(x_{\cdot},y_{\cdot}\right)$ is closed,
convex and bounded for each $\left(x_{\cdot},y_{\cdot}\right)\in C\left(\left[0,T\right];\mathbb{R}^{2}\right)$
and each $t\geqslant0$ 
\end{lemma}
\begin{proof} Only convexity needs to be checked.
Let $0\leqslant\mu\leqslant1$ and $\left(a,b\right)$,
$\left(\alpha,\beta\right)\in A_{t}\left(x_{\cdot},y_{\cdot}\right)$
then the convex linear combination $\mu a+(1-\mu)\alpha$ is equal to 
\[
\frac{1}{2}\kappa^2\left(t,y_{\cdot}\right)E+\frac{1}{2}\left(\mu m +\left(1-\mu\right)m'\right) 
\lambda^2\left(t,y_{\cdot}\right)x^{2}_t\left(\begin{array}{cc} \rho^2\left(t,x_{\cdot}\right) & \rho\left(t,x_{\cdot}\right) \\ \rho\left(t,x_{\cdot}\right) & 1\end{array}\right)
\]
and $\mu b+(1-\mu)\beta$ equals
\[
\left(\begin{array}{c}
         \nu\left(t,y_{\cdot}\right) \\ 0 
  \end{array} \right) 
  + \left(\left(\lambda l+\left(1-\lambda\right)l'\right) x_t\theta^r\left(t,y_{\cdot}\right)+r_t x_t\right)\left(\begin{array}{c} \rho\left(t,x_{\cdot}\right) \\ 1 \end{array} \right)
\]
As $\mu l+(1-\mu)l'\leqslant\mu\sqrt{m}+(1-\mu)\sqrt{m'}\leqslant\sqrt{\mu m+(1-\mu)m'}$,we have
$\mu\left(a,b\right)+(1-\mu)\left(\alpha,\beta\right)\in A_{t}\left(x_{\cdot},y_{\cdot}\right)$.
\end{proof}

The estimates of Lemma \ref{rcoef} apply to this case as well, for some $K>0$,
\[
\left\| A_{t}\left(x_{\cdot},y_{\cdot}\right)\right\| \leqslant K\left(1+\left|X_{t}\right|^{2}\right).
\]
This allows to apply the results of \cite{Kry01} just as above, using the class of
relaxed controls defined below.


\begin{definition}\label{bPi2}
 We say that $\overline{\pi}\in\overline{\Pi}$ if 
 $$
\overline{\pi}:=\left(\Omega,\mathcal{F},\left\{\mathcal{F}_t\right\}_{0\leqslant t\leqslant T},
 \mathbb{P},X_t,Y_t,\left(B_t,W_t\right),(x,y)\right)$$ 
 with
\begin{description}
\item[(a)] $\left(\Omega,\mathcal{F},\left\{\mathcal{F}_t\right\}_{0\leqslant t\leqslant T},
\mathbb{P}\right)$ a complete filtered probability space  whose filtration satisfies the usual conditions;
\item[(b)] the $2$-dimensional process $\xi_t:=\left(B_t,W_t\right)$ is a standard $\mathcal{F}_t$-Brownian motion; 
\item[(c)] the vector $(x,y)\in (0,\infty)\times\mathbb{R}$ is the initial endowment of the portfolio 
process $X_t$ and the initial state of the economic factors $Y_t$, respectively;  
\item[(d)] there exists an $\mathbb{A}$-valued, $\mathcal{F}\otimes\mathcal{B}\left(\left[0,T\right]\right)$ measurable and $\mathcal{F}_t$-adapted process denoted by $\left(a_t,b_t\right)$ such that  
\begin{equation}
\left(\begin{array}{c} Y_t  \\  X_t \\ \end{array}\right)=\int_0^t \sqrt{2 a_s} d\xi_s+ \int_0^t b_s ds
\end{equation}


\item[(e)] for almost all $\left(\omega,t\right)\in\Omega\times\left[0,T\right]$, we have 
$\left(a_t,b_t\right)\in A_t\left(X_{\cdot},Y_{\cdot}\right)$ (i.e. we can choose a pair $(m_t,l_t)$ in a 
``measurable way'').
\end{description}
\end{definition}

The vectorial form of the equations (\ref{eq1}) and (\ref{eq2}) can be rewritten. Define
\[
\sigma_t:=\left(\begin{array}{cc} \kappa\left(t,y_{\cdot}\right) & \rho\left(t,x_{\cdot}\right)\sqrt{m_t}x_t\lambda\left(t,y_{\cdot}\right)  \\ 0 & \sqrt{m_t}_t\lambda\left(t,y_{\cdot}\right)x_t \\ \end{array} \right)\,
\]
and we have that the drift is given by
\[
b_t=\left(\begin{array}{c} \nu\left(t,y_{\cdot}\right)\\ 0 \\ \end{array}\right)+\left(x_tl_t\theta^r\left(t,Y_{\cdot}\right)+x_tr_t\right)\cdot\left(\begin{array}{c} \rho\left(t,x_{\cdot}\right)    \\  1 \\ \end{array}\right)
\]

Given a relaxed control $\overline{\pi}$,
$X_t,Y_t$ are $\mathcal{F}\otimes\mathcal{B}\left(\left[0,T\right]\right)$-measurable and $\mathcal{F}_t$-adapted such that for all $t\geqslant 0$
\begin{equation}\label{eq5}
dY_t=\nu\left(t,Y_{\cdot}\right)dt+\kappa\left(t,Y_{\cdot}\right)dB_t + \rho\left(t,X_{\cdot}\right)dX_t,
\end{equation}
\begin{equation}\label{eq6}
dX_t=\left[l_t\left(\theta\left(t,Y_{\cdot}\right)-r_t\right)X_t+r_t\cdot X_t\right]dt+\sqrt{m_t}\lambda\left(t,Y_{\cdot}\right)X_tdW_t. 
\end{equation}

The proof of the next result follows closely that of Theorem \ref{wp1}.

\begin{theorem}\label{opt2}
Let Assumptions \ref{ul}, \ref{u}, \ref{inter} and \ref{ghyp2} hold. The problem \eqref{eq:V funt} is well-posed
and $\Pi'\neq\emptyset$ (the identically zero strategy belongs to $\Pi'$, where $\Pi'$ is defined
analogously to Subsection \ref{mrr}).
There is $\hat{\pi}\in\Pi'$ such that the supremum in 
(\ref{eq:V funt}) is attained. \hfill $\Box$
\end{theorem}

\section{Appendix}
Some proofs of auxiliary results are included in this section.

\begin{proof}[Proof of Proposition \ref{mestim1}]
We shall write $\xi_s=\left(W_s,B_s\right)$.  
Suppose $m\geqslant2$. The notation $|\cdot|$ will be used to denote Euclidean norm in spaces of
various dimensions. Then  
\begin{equation}
|\zeta_{t}|^{m}=\left[\left(X_{t}\right)^{2}+\left|Y_{t}\right|^{2}\right]^{\frac{m}{2}}
\leqslant2^{\frac{m}{2}-1}\cdot\left[|X_{t}|^{m}+\left|Y_{t}\right|^{m}\right],
\end{equation}
so it is enough to obtain that the moments of each of the processes $Y_{t}$ and $X_{t}$ satisfy (\ref{totbnd}). 
Set 
$b_{s,2}=\left(l_{t}\theta_{t}\left(Y_{\cdot}\right)X_{t}\right)$ and $\sigma_{s,2}=\left(0,m_{t}\lambda_{t}\left(Y_{\cdot}\right)X_{t}\right)$
\begin{equation}
\mathbf{E}_{\pi}\left[\sup_{t\leqslant T}\left|X_{t}\right|^{m}\right]\leqslant3^{m-1}
\left(\left|X_{0}\right|^{m}+\mathbf{E}_{\pi}\left(\int_{0}^{T}|b_{s,2}|ds\right)^{m}+\mathbf{E}_{\pi}\left[\sup_{t\leqslant T}\left|\int_{0}^{t}\sigma_{s,2}d\xi_s\right|^{m}\right]\right),
\end{equation}

By Jensen's inequality and Burkholder-Davis-Gundy inequality 
\[
\mathbf{E}_{\pi}\left[\sup_{t\leqslant T}\left|X_{t}\right|^{m}\right]\leqslant3^{m-1}\left(\left|X_{0}\right|^{m}+\mathbf{E}_{\pi}\left(\int_{0}^{T}|b_{s,2}|ds\right)^{m}+C_{m}\mathbf{E}_{\pi}\left[\left|\int_{0}^{T}\left|\sigma_{s,2}\right|^{2}ds\right|^{\frac{m}{2}}\right]\right),
\]
for some $C_m>0$.

We can apply again Jensen's inequality (now with respect the ``uniform density'' on $\left[0,T\right]$)
\begin{equation}
\left(T\int_{0}^{T}\frac{|b_{s,2}|}{T}ds\right)^{m}\leqslant T^{m-1}\cdot\int_{0}^{T}|b_{s,2}|^{m}ds,\,\,\,\mbox{and }\,\,\left|T\int_{0}^{T}\frac{\left\|\sigma_{s,2}\right\|^{2}}{T}ds\right|^{\frac{m}{2}}\leqslant T^{m/2-1}\cdot\int_{0}^{T}\left\|\sigma_{s,2}\right\|^{m}ds,
\end{equation}
here $\left\|\sigma_{s,2}\right\|^{m}=m_{t}^{m}\left|\lambda_{t}\left(Y_{\cdot}\right)\right|^{m}\cdot X_{t}^{m}$ and $|b_{s,2}|^{m}=l_{t}^{m}\theta_{t}^{m}\left(Y_{\cdot}\right)X_{t}^{m}$, however, we can use the estimate 
(\ref{bndeq}) above,
\[
\mathbf{E}_{\pi}\left[\sup_{t\leqslant T}\left|X_{t}\right|^{m}\right]\leqslant 3^{m-1}\left(\left|X_{0}\right|^{m}+K\mathbf{E}_{\pi}\left[\int_{0}^{T}\left(1+\sup_{t\leqslant s}\|\zeta_{t}\|\right)^{m}ds\right]\right),
\]
similarly,
\[
\mathbf{E}_{\pi}\left[\sup_{t\leqslant T}\left|Y_{t}\right|^{m}\right]\leqslant3^{m-1}\left(\left|Y_{0}\right|^{m}+K' \mathbf{E}_{\pi}\left[\int_{0}^{T}\left(1+\sup_{t\leqslant s}\|\zeta_{t}\|\right)^{m}ds\right]\right),
\]
for constants $K,K'$.
Then we have 
\begin{equation}
\mathbf{E}_{\pi}\left[\sup_{t\leqslant T}\|\zeta_{t}\|^{m}\right]\leqslant K(m)\left(\|\zeta_{0}\|^{m}+\left[\int_{0}^{T}1+\mathbf{E}_{\pi}\left(\sup_{t\leqslant s}\|\zeta_{t}\|\right)^{m}ds\right]\right),\label{eq:growns}
\end{equation}
for some $K(m)>0$ so by Gronwall's lemma, 
\[
\mathbf{E}_{\pi}\left[\sup_{t\leqslant T}\|\zeta_{t}\|^{m}\right]\leqslant
L(m),
\]
with a fixed constant $L(m)$, for all $\pi\in\Pi$.
The case $0<m<2$ follows from the monotonicity of the norms. 
\end{proof}

\begin{proof}[Proof of Proposition \ref{tgh}]
As in Proposition \ref{mestim1}, it is enough to show a similar estimate (\ref{eq:tightnessK}) for each of the coordinates $X_{t}$ and $Y_{t}$. That means 
\begin{equation}
\mathbf{E}\left|Y_{t}-Y_{s}\right|^{\eta}\leqslant K_{1}\left|t-s\right|^{\eta/2}\,\,\,\mbox{and }\,\,\,\mathbf{E}\left|X_{t}-X_{s}\right|^{\eta}\leqslant K_{2}\left|t-s\right|^{\eta/2}.
\end{equation}
By Assumption \ref{acoef1} the first inequality is a simple consequence of B-D-G's and Jensen's inequality: 
\[
\mathbf{E}\left|Y_{t}-Y_{s}\right|^{\eta}\leqslant2^{\eta-1}\cdot\left[\mathbf{E}\left(\int_{s}^{t}\left|\nu\left(Y_{[0,r]}\right)\right|dr\right)^{\eta}+\mathbf{E}\left(\sup_{s\leqslant r\leqslant t}\left|\int_{s}^{r}\kappa_{r}\left(Y_{[0,r]}\right)dB_{r}\right|^{\eta}\right)\right],
\]
\[
\mathbf{E}\left|Y_{t}-Y_{s}\right|^{\eta}\leqslant2^{\eta-1}\cdot\left[M^{\eta}\left|t-s\right|^{\eta}+C_{\eta}\cdot\mathbf{E}\left(\int_{s}^{t}\kappa_{r}^{2}\left(Y_{[0,r]}\right)dr\right)^{\eta/2}\right],
\]
thus
\[
\mathbf{E}\left|Y_{t}-Y_{s}\right|^{\eta}\leqslant2^{\eta-1}\cdot\left[M^{\eta}\left|t-s\right|^{\eta}+C_{\eta}M^{\eta}\left|t-s\right|^{\eta/2}\right]\leqslant 2^{\eta-1}\cdot\left[M^{\eta}T^{\eta/2}+C_{\eta}M^{\eta}\right]\cdot\left|t-s\right|^{\eta/2}.
\]
 The second estimate relies upon Proposition \ref{mestim1}: 
\[
\mathbf{E}\left|X_{t}-X_{s}\right|^{\eta}\leqslant2^{\eta-1}\cdot\left[\mathbf{E}\left(\int_{s}^{t}\left|l_{r}\theta_{r}\left(Y_{[0,r]}\right)X_{r}\right|dr\right)^{\eta}+\mathbf{E}\left(\sup_{s\leqslant r\leqslant t}\left|\int_{s}^{r}\lambda_{r}\left(Y_{[0,r]}\right)\sqrt{m_{r}}\, X_{r}dW_{r}\right|^{\eta}\right)\right],
\]
\[
\leqslant2^{\eta-1}\cdot M^{\eta}\cdot\left[\mathbf{E}\left(\int_{s}^{t}\left|X_{r}\right|dr\right)^{\eta}+C_{\eta}\mathbf{E}\left(\left|\int_{s}^{t}\,\left|X_{r}\right|^{2}dr\right|^{\eta/2}\right)\right],
\]
\[
\mathbf{E}\left|X_{t}-X_{s}\right|^{\eta}\leqslant 2^{\eta-1}\cdot M^{\eta}\left[\left(t-s\right)^{\eta}\cdot\mathcal{N}\left(\eta,T\right)+\left(t-s\right)^{\eta/2}\cdot\mathcal{N}\left(\eta,T\right)\right]= K_{3}\left|t-s\right|^{\eta/2},
\]
where $K_{3}=2^{\eta-1}M^{\eta}\mathcal{N}\left(\eta,T\right)\cdot\left(T^{\eta/2}+1\right)$
and $\mathcal{N}\left(\eta,T\right)$ is the upper bound of 
$\sup_{\overline{\pi}\in\overline{\Pi}}\mathbf{E}_{\pi}\left[\sup_{t\leqslant T}|X_{t}|^{\eta}\right]$ 
in Proposition \ref{mestim1}. Note that the constants do not depend on $\overline{\pi}$ as neither $M$ nor 
$\mathcal{N}\left(\eta,T\right)$ do. 
\end{proof}


\begin{thebibliography}{99}

\bibitem{BEKSY} M. Barlow, M. \'Emery, F. Knight, S. Song, and M. Yor.
\newblock\emph{Autour d'un th\'eoreme de Tsirelson sur des filtrations browniennes et non browniennes},
\newblock{In: S\'eminaire de Probabilit\'es, XXXII, Lecture Notes in Math.} 1686, 264--305, Springer,
Berlin, 1998.  


\bibitem{bernoulli} D. Bernoulli.
\newblock Theoriae Novae de Mensura Sortis.
\newblock \emph{Commentarii Academiae Scientiarum Imperialis Petropolitanae. Volume V.}, 1738.
\newblock Translated by L. Sommer as ``Exposition of a New Theory on the Measurement of Risk'',
\emph{Econometrica}, 22:23--36, 1954.


\bibitem{bkp}
A.~B. Berkelaar, R.~Kouwenberg, and T.~Post.
\newblock Optimal portfolio choice under loss aversion.
\newblock \emph{Rev. Econ. Stat.}, 86:\penalty0 973--987, 2004.


\bibitem{black}
F.~Black.
\newblock Studies of stock market volatility changes.
\newblock {\em Proceedings of the American Statistical Association, Business
 and Economic Statistics Section}, 177--181, 1976.

\bibitem{campi}
L.~Campi and M.~Del~Vigna.
\newblock Weak insider trading and behavioural finance.
\newblock \emph{SIAM J. Financial Mathematics}, 3:\penalty0 242--279, 2012.


\bibitem{CarRas11} L. Carassus and M. R\'asonyi.
\newblock{On optimal investment for a behavioral investor in multiperiod incomplete markets},
\newblock\emph{To appear in Math. Finance,} 2013. Available at \texttt{http://arxiv.org/abs/1107.1617}


\bibitem{cd}
G.~Carlier and R.-A. Dana.
\newblock Optimal demand for contingent claims when agents have law invariant
  utilities.
\newblock \emph{Math. Finance}, 21:\penalty0 169--201, 2011.





\bibitem{cherny}
A. Cherny and D. P. Madan. \newblock New Measures for Performance Evaluation.
\newblock \emph{Review of Financial Studies}, 22:2571--2606, 2009.

\bibitem{papa} J.-P. Fouque, G. Papanicolau and R. Sircar. \newblock\emph{Derivatives in Financial Markets 
with Stochastic Volatility.} \newblock Cambridge University Press, 2000.



\bibitem{jz}
H.~Jin and X.~Y. Zhou.
\newblock Behavioural portfolio selection in continuous time.
\newblock \emph{Math. Finance}, 18:385--426, 2008.






\bibitem{kt}
D.~Kahneman and A.~Tversky.
\newblock Prospect theory: An analysis of decision under risk.
\newblock \emph{Econometrica}, 47:\penalty0 263--291, 1979.




\bibitem{Kry01} N. V. Krylov.
\newblock{A supermartingale characterization of sets of stochastic integrals and applications},
\newblock\emph{Probab. Theory Relat. Fields,} {123}:521--552, 2002.

\bibitem{krli} Krylov N. V., Liptser R.
\newblock\emph{On diffusion approximation with discontinuous coefficients},
\newblock{Stoch. Proc. Appl.} {102}:235--264 (2002).



\bibitem{KT79} D. Kahneman and A. Tversky.
\newblock {Prospect theory: An analysis of decision under risk},  
\emph{Econometrica} {47}:263--292, 1979. 




\bibitem{quiggin} J. Quiggin. \newblock A Theory of Anticipated Utility. 
\newblock\emph{Journal of Economic and Behavioral
Organization,} 3:323--343, 1982.


\bibitem{RR11} M. R\'{a}sonyi and A. M. Rodrigues.
\newblock{Optimal portfolio choice for a behavioural investor in continuous-time markets},
\newblock\emph{Ann. Finance} {9}:291--318, 2013. 


\bibitem{r12}
C.~Reichlin.
\newblock Utility maximization with a given pricing measure when the utility is
  not necessarily concave.
\newblock \emph{Mathematics and Financial Economics}, 7:\penalty0 531--556,
  2013.


\bibitem{reichlinthesis}
C.~Reichlin, \emph{Non-concave utility maximization: optimal investment,
  stability and applications.} 
  \newblock Ph.D.\ thesis, ETH Z\"{u}rich, 2012. 
  \newblock {Diss. ETH No. 20749}.
  



\bibitem{tk}
A. Tversky and D. Kahneman.
\newblock {Advances in prospect theory: Cumulative representation of
  uncertainty},
\emph{J. Risk Uncertainty}, 5:297--323, 1992.

\bibitem{neumann53}
J.~von Neumann and O.~Morgenstern, \emph{Theory of games and economic
  behavior}, Princeton University Press, 1944.



\end{thebibliography}
\end{document}